\documentclass{amsart}

\usepackage{tikz}
\usetikzlibrary{positioning}

\usepackage{mathtools}
\newcommand{\defeq}{\vcentcolon=}

\newtheorem{theorem}{Theorem}[section]
\newtheorem{lemma}[theorem]{Lemma}

\theoremstyle{definition}
\newtheorem{definition}[theorem]{Definition}
\newtheorem{example}[theorem]{Example}
\newtheorem{corollary}[theorem]{Corollary}

\theoremstyle{remark}
\newtheorem{remark}[theorem]{Remark}

\numberwithin{equation}{section}

\begin{document}

\title{Metrics on Permutation Families Defined by a Restriction Graph}


\author{Danylo Tymoshenko}
\address{Department of Mathematics and Statistics, Washington University in St. Louis}
\email{d.tymoshenko@wustl.edu}

\author{Leonhard Nagel}
\address{Department of Electrical Engineering and Computer Sciences, University of California, Berkeley}
\email{nagel@berkeley.edu}


\date{July 2024}

\dedicatory{}

\begin{abstract}
Understanding the metric structure of permutation families is fundamental to combinatorics and has applications in social choice theory, bioinformatics, and coding theory. We study permutation families defined by restriction graphs—oriented graphs that constrain the relative order of elements in valid permutations.

For any restriction graph $G$, we determine the maximum distance achievable by two permutations under the $\ell_\infty$-metric and provide an explicit algorithm that constructs optimal permutation pairs. Our main contribution characterizes when the Kendall-Tau metric achieves its combinatorial upper bound: this occurs if and only if the poset induced by $G$ has dimension at most 2. When this condition holds, the extremal permutations form a minimal realizer of the poset, revealing a deep connection between metric geometry and poset dimension theory.

We apply these results to classical permutation statistics including descent sets and Hessenberg varieties, obtaining explicit formulas and efficient algorithms for computing metric diameters.
\end{abstract}

\maketitle

\section{Introduction}

The study of permutation families and their metric properties has applications spanning from coding theory to computational biology and social choice theory. In coding theory, permutation codes with specific distance properties are used for error correction in power line communications. In bioinformatics, understanding distances between gene arrangements helps reconstruct phylogenetic trees. In social choice, metric properties of preference orderings inform voting system design.

This paper introduces and analyzes \emph{restriction graphs}—a unifying framework for studying permutation families defined by ordering constraints. A restriction graph $G$ on vertices $[n] = \{1,2,\ldots,n\}$ specifies which pairs of positions must satisfy $\sigma_i > \sigma_j$ in any valid permutation $\sigma$. This framework generalizes classical combinatorial objects including descent sets, peak sets, and inversion patterns studied extensively in algebraic combinatorics.

Our main contributions determine the maximum distances achievable within these permutation families under two important metrics: the $\ell_\infty$-metric (maximum coordinate difference) and the Kendall-Tau metric (number of pairwise disagreements). These results have both theoretical significance—revealing connections between metric geometry and poset dimension theory—and practical applications in algorithm design for permutation enumeration and optimization.

We will write $[n]$ to denote the set $\{1,2,3,\ldots,n\}$ and $S_n$ for the symmetric group on $[n]$. Throughout this paper, all graphs have vertex set $[n]$ with vertices labeled by their natural order. We write permutations $\sigma \in S_n$ as $\sigma=\sigma_1 \sigma_2\ldots\sigma_n$ where $\sigma_i$ denotes the element at position $i$.

\begin{center}
\textbf{Notation Summary}
\end{center}
\begin{tabular}{ll}
$[n]$ & The set $\{1,2,\ldots,n\}$ \\
$S_n$ & Symmetric group on $[n]$ \\
$G=([n],E)$ & Oriented restriction graph with vertex set $[n]$ and edge set $E$ \\
$\mathcal{P}(G)$ & Set of all permutations satisfying restriction graph $G$ \\
$u \rightsquigarrow v$ & There exists an oriented path from $u$ to $v$ in $G$ \\
$R(v)$ & Reachable set: $\{u \neq v : v \rightsquigarrow u\}$ \\
$R^{-1}(v)$ & Inverse reachable set: $\{u \neq v : u \rightsquigarrow v\}$ \\
$P = ([n], \leq_G)$ & Poset formed by restriction graph $G$ \\
$\operatorname{Incomp}(P)$ & Set of incomparable pairs in poset $P$ \\
$d_\ell(\sigma,\rho)$ & $\ell_\infty$-metric: $\max_{1 \leq i \leq n} |\sigma_i-\rho_i|$ \\
$d_K(\sigma,\rho)$ & Kendall-Tau metric: number of discordant pairs $(i,j)$ \\
$\dim P$ & Dimension of poset $P$ \\
$\mathcal{D}(\sigma)$ & Descent set of permutation $\sigma$ \\
\end{tabular}

\begin{definition}
    An oriented graph $G=([n], E)$ with vertex set $[n]$ and directed edge set $E$ is called a \textbf{restriction graph} for a permutation $\sigma \in S_n$ 
    if for every oriented edge $u \to v \in E$ we have $\sigma_u > \sigma_v$.

    If a permutation satisfies this property for $G$, we say that it \textbf{satisfies} the restriction graph $G$. 
    The set of all permutations satisfying $G$ is denoted as $\mathcal{P}(G)$.
\end{definition}

\begin{example} \label{ex:1} 
Let restriction graph $G$ be represented as:
    
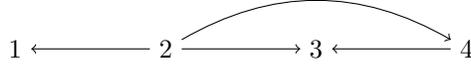
\begin{center}
    \begin{tikzpicture}[node distance=2cm]
        \node (1) {1};
        \node (2) [right of=1] {2};
        \node (3) [right of=2] {3};
        \node (4) [right of=3] {4};
        
        \draw[->] (2) -- (1);
        \draw[->] (2) -- (3);
        \draw[->] (4) -- (3);
        \draw[->] (2) to[bend left=30] (4);
    \end{tikzpicture}
    \end{center}
    
    This means permutation $\sigma$ satisfies $G$ if and only if:
    \[ \sigma_2 > \sigma_1, \quad \sigma_2 > \sigma_3, \quad \sigma_4 > \sigma_3, \quad \sigma_2>\sigma_4 \]
    
    The set of all permutations satisfying $G$ is $\mathcal{P}(G) = \{1423, 2413, 3412\}$.
\end{example}

\begin{definition}
    We say that graph $G$ is \textbf{$k-$admissible} if $\mathcal{P}(G)$ contains at least $k$ elements.

    For simplicity of writing, in this paper, if we say $G$ is a restriction graph, we will mean that it is $1-$admissible. Notice that saying $G$ is $1-$admissible only means that there is \emph{at least} one permutation satisfying it; there \emph{can be} two or more permutations satisfying a $1-$admissible graph.
\end{definition}
\begin{example}
    The graph $G$ from Example \ref{ex:1} is $3-$admissible. This implies that $G$ is a restriction graph since it is also $1-$admissible then.
\end{example}

\begin{example}
    Consider the following graph:
    
    
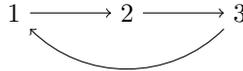
\begin{center}
    \begin{tikzpicture}[node distance=1.5cm]
        \node (1) {1};
        \node (2) [right of=1] {2};
        \node (3) [right of=2] {3};
        
        \draw[->] (1) -- (2);
        \draw[->] (2) -- (3);
        \draw[->] (3) to[bend left=45] (1);
    \end{tikzpicture}
    \end{center}
    Notice that this graph implies that permutation $\sigma$ satisfying it should follow the following inequalities:
    \[ \sigma_1>\sigma_2, \quad \sigma_2>\sigma_3, \quad \sigma_3>\sigma_1 \]
    But then
    \[ \sigma_1>\sigma_2>\sigma_3>\sigma_1 \] 
    \\
    The last inequality is a contradiction, showing that there cannot be any permutation satisfying this graph. So this graph \emph{is not} a restriction graph since there is no permutation that could satisfy it.

    In this example, a contradiction was obtained by considering an oriented cycle in the graph. 
    This naturally leads to the idea that restriction graphs should not have any oriented cycles. 
    We establish this result in the next two theorems.
\end{example}

\begin{theorem} \label{theorem: path}
    If there is an oriented path from $u$ to $v$ in a restriction graph $G$, then $\sigma_u>\sigma_v$ for any $\sigma$ satisfying $G$.
\end{theorem}

\begin{proof} 
    Let $(u, w_1, \dots, w_k , v)$ be an oriented path from $u$ to $v$ in oriented graph $G$. Then the existence of edges $uw_1,w_1w_2\dots,w_{k-1}w_k,w_kv$ implies that
    \[ \sigma_u>\sigma_{w_1},\quad \sigma_{w_1}>\sigma_{w_2}, \quad \dots, \quad \sigma_{w_{k-1}}>\sigma_{w_k},\quad \sigma_{w_k}>\sigma_v \quad \text{ for any } \sigma \in \mathcal{P}(G) \] 
    which further simplifies into
    \[ \sigma_u>\sigma_{w_1}>\sigma_{w_2}>  \dots>\sigma_{w_{k-1}}>\sigma_{w_k}>\sigma_v \text{ for any } \sigma \in \mathcal{P}(G) \] 
    \[ \sigma_u>\sigma_v \text{ for any } \sigma \in \mathcal{P}(G)  \]
\end{proof}

\begin{theorem} \label{theorem:DAG}
    Any restriction graph with $|\mathcal{P}(G)| \geq 1$ (i.e., any $1$-admissible restriction graph) must be acyclic.
\end{theorem}
\begin{proof}
    Assume, for contradiction, that there exists a $1$-admissible restriction graph $G$ 
    with an oriented cycle $(u_1 \to u_2 \to \cdots \to u_k \to u_1)$. 
    Since $G$ is $1$-admissible, there exists a permutation $\sigma$ satisfying it. 
    By Theorem \ref{theorem: path}, the oriented path $u_1 \to u_2 \to \cdots \to u_k \to u_1$ implies 
    $\sigma_{u_1} > \sigma_{u_2} > \cdots > \sigma_{u_k} > \sigma_{u_1}$, which gives $\sigma_{u_1} > \sigma_{u_1}$, 
    a contradiction.
\end{proof}

The introduced $\mathcal{P}(G)$ is a generalization of $\mathcal{D}(S;n)$, the set of all permutations having their descent set equal $S$, studied earlier by Alexander Diaz-Lopez et al. \cite{DHMM24}. 
In their paper, they found the maximum value of the $\ell_\infty$-metric metric for some special cases of $\mathcal{D}(S;n)$. 
To the best of our knowledge, this paper provides the first systematic study of metrics on $\mathcal{P}(G)$ for general restriction graphs, while also showing a simpler approach for any $\mathcal{D}(S;n)$. 
Additionally, this paper discovers the maximum value of another metric --- the Kendall-Tau metric for arbitrary restriction graph $G$ and special cases of it, which was earlier only discovered for the peak sets \cite{DHKPW24}. A key result is our complete characterization: the maximum Kendall-Tau distance on $\mathcal{P}(G)$ equals the number of incomparable pairs in the associated poset if and only if the poset has dimension at most 2.

\section{The $\ell_\infty$-Metric: Maximum Coordinate Differences}

\subsection{Definitions and main results}

\begin{definition}
    The \textbf{$\ell_\infty$-metric} is a map $d_\ell: S_n \times S_n \rightarrow[0; \infty)$ such that 
    \[ d_\ell(\sigma, \rho) = \max_{1 \leq i \leq n} |\sigma_i-\rho_i| \]
\end{definition}

\begin{definition}
    If there is an oriented path from $u$ to $v$ in graph $G$, we write $u \rightsquigarrow v$, and $u \not\rightsquigarrow v$ otherwise.
\end{definition}

\begin{definition}
    In an oriented graph $G=(V,E)$ the \textbf{reachable set} of vertex $v$ is denoted as $\boldsymbol{R(v)}$ 
    and is defined as the set of all vertices $u \neq v$ such that there is an oriented path from $v$ to $u$:
    \[ R(v) \defeq \{ u \in V \setminus \{v\} \mid v \rightsquigarrow u \} \]
    The number of vertices reachable from $v$ is $\boldsymbol{|R(v)|}$. Note that $v$ itself is excluded from $R(v)$.
\end{definition}

\begin{remark}
    The exclusion of $v$ from $R(v)$ ensures that for any permutation $\sigma \in \mathcal{P}(G)$, the value $\sigma_v$ lies in the tight interval $[|R(v)|+1, n-|R^{-1}(v)|]$, which is essential for our metric bound analysis.
\end{remark}

\begin{definition}
    In an oriented graph $G=(V,E)$ the \textbf{inverse reachable set} of vertex $v$ is denoted as $\boldsymbol{R^{-1}(v)}$ 
    and is defined as the set of all vertices $u \neq v$ such that there is an oriented path from $u$ to $v$:
    \[ R^{-1}(v)\defeq\{ u \in V \setminus \{v\} \mid u \rightsquigarrow v \} \]
    The number of vertices from which $v$ can be reached is $\boldsymbol{|R^{-1}(v)|}$. Again, $v$ itself is excluded.
\end{definition}

\subsection{Computing maximum distances}

Note that in this and all subsequent sections discussing metrics, we require restriction graphs to be at least $2$-admissible (i.e., $|\mathcal{P}(G)| \geq 2$), since metrics are only meaningful between distinct elements.

\begin{theorem}
    The maximum value of the $\ell_\infty$-metric on $\mathcal{P}(G)$ satisfies 
    \begin{equation*}
        \max\{ d_\ell(\mathcal{P}(G)) \} \leq  \max_{1 \leq i \leq n} \{n-|R(i)|-|R^{-1}(i)|-1 \}
    \end{equation*}
\end{theorem}

\begin{proof}
    Consider an arbitrary restriction graph $G$ and a permutation $\sigma$ satisfying it. Let $i \in [n]$ be an arbitrary vertex of $G$, to which $\sigma_i$ corresponds.
    
    Notice that for every $j\in R(i)$, since there is a directed path from $i$ to $j$ by Theorem \ref{theorem: path} we have $\sigma_i>\sigma_j$.
    Then $\sigma_i$ is greater than all the elements corresponding to the indices in $R(i)$. 
    Since all the elements corresponding to the indices defined by $R(i)$ are distinct, we have $\sigma_i \geq |R(i)|+1$.
    
    On the other hand for every $j\in R^{-1}(i)$ since there is a directed path from $j$ to $i$  by Theorem \ref{theorem: path} we have $\sigma_i<\sigma_j$.
    Then $\sigma_i$ is less than all the elements corresponding to the indices in $R^{-1}(i)$. 
    Since all the elements corresponding to the indices defined by $R^{-1}(i)$ are distinct, we have $\sigma_i \leq n-|R^{-1}(i)|$.

    Combining the two discovered boundaries we have $\sigma_i \in [|R(i)|+1, n-|R^{-1}(i)|]$.
    Then the maximum possible difference between $\sigma_i$ and $\rho_i$ for two permutations $\sigma, \rho$ satisfying the restrictions of $G$ is 
    \[ n-|R^{-1}(i)| -(|R(i)|+1) = n-|R(i)|-|R^{-1}(i)|-1 \]

    Now we can notice that the maximum possible difference formula holds for all $i\in [n]$, so if we consider index $k$ where the value of 
    $|R(k)|+|R^{-1}(k)|$ is minimal, then this position will give the maximum difference of all. 
    Since minimizing $|R(k)|+|R^{-1}(k)|$ is the same as maximizing $n-|R(i)|-|R^{-1}(i)|-1$ we get that

    \[ d_\ell(\sigma, \rho) \leq \max_{1 \leq i \leq n} \{n-|R(i)|-|R^{-1}(i)|-1 \} \]
    However, since the right side doesn't depend on the choice of $\sigma$ and $\rho$, we get
    \[\max\{ d_\ell(\mathcal{P}(G)) \} \leq  \max_{1 \leq i \leq n} \{n-|R(i)|-|R^{-1}(i)|-1 \}\]
\end{proof}

Now, since we proved that the maximum value cannot exceed a certain boundary, it is sufficient to show that this boundary can always be achieved. To do that, we first need to introduce an algorithm. The algorithm will take a restriction graph $G$ and create a permutation $\sigma$ that satisfies it.

\begin{center}
\textbf{Algorithm: Construction of permutation satisfying restriction graph $G$}
\end{center}
\begin{enumerate}
    \item Initialize set $X \leftarrow [n]$.
    \item While $X \neq \emptyset$:
    \begin{enumerate}
        \item Find the leftmost source $s$ in the current graph (i.e., the source with the smallest label).
        \item Set $\sigma_s \leftarrow \max(X)$.
        \item Remove $\max(X)$ from $X$.
        \item Remove vertex $s$ and all its incident edges from $G$.
    \end{enumerate}
    \item Return $\sigma$.
\end{enumerate}

\begin{lemma}
    The permutation $\sigma$ returned by this algorithm satisfies graph $G$
\end{lemma}
\begin{proof}
    We prove this by induction on the size of the graph.

    \textbf{Induction claim.} 
For any $n$ the algorithm generates a permutation which satisfies restriction graph $G$.

\textbf{Induction basis.} 
Consider $n=1$. In this case, the only possible $G$ to consider is $G=(\{1\},\emptyset)$.
Then, the algorithm just puts $1$ at index 1. Permutation $\sigma=1$ clearly satisfies $G$ since it does not apply any restrictions on it.

\textbf{Induction step.} 
Assume that the claim holds for $n$. Consider restriction graph $G$ on $n+1$ vertices.

By Theorem \ref{theorem:DAG} $G$ is acyclic, which means it should have at least one source. Call the index of the leftmost source $s$. By the algorithm, we set vertex $s$ to be equal to $n$, and since it is the biggest possible value in a permutation on $n$ elements, all edges coming out of it are satisfied. 

Now make a bijection $f:[n] \setminus\{s\}\to[n-1]$ that maps all the values in the set $[s-1]$ to itself and all the values in the set $[n]\setminus [s]$ to itself decreased by $1$. Then, we can map graph $G$ to a new graph $G'$, deleting all the edges coming from vertex $s$ and mapping vertices by applying $f$ to them. Using the algorithm on $G'$ gives us a permutation satisfying $G'$. 

Since all relationships between vertices in $[n]\setminus \{s\}$ are guaranteed to be satisfied by the induction assumption, then we only need to check that all the relationships with $s$ are satisfied. But they were satisfied by the construction, which ends the proof.

\end{proof}

\begin{theorem} \label{theorem:linfThm}
    For any restriction $2-$admissible graph $G$ there are two permutations $\sigma,\rho \in \mathcal{P}(G)$ such that $d_\ell(\sigma, \rho)=\max_{1 \leq i \leq n} \{n-|R(i)|-|R^{-1}(i)|-1 \}$
\end{theorem}

\begin{proof}
    Let $k \in [n]$ be a position at which the value of $|R(k)|+|R^{-1}(k)|$ is minimal.
    We will construct two permutation $\sigma \text{ and } \rho$ such that $d_\ell(\sigma, \rho)=n-|R(k)|-|R^{-1}(k)|-1$.
    Set $\sigma_k=|R(k)|+1$ and $\rho_k=n-|R^{-1}(k)|$.
    Notice that $\rho_k-\sigma_k=n-|R(k)|-|R^{-1}(k)|-1$ and that $\sigma_k,\rho_k \in [|R(k)|+1, n-|R^{-1}(k)|]$, which is the interval of possible values.
    
    Now consider induced subgraph $G[R(k)]$.
    We apply the algorithm to both $\sigma$ and $\rho$ on this induced subgraph.
    The algorithm will set all the values from the interval $[1, |R(k)|]$ in these positions in the way that all the restrictions between them hold.
    
    Consider induced subgraph $G[R^{-1}(k)]$.
    We will apply the algorithm on this subgraph for the set of values $[n-|R^{-1}(k)|+1, n]$.

    Then finally consider induced subgraph $G[[n] \setminus (R(k) \cup R^{-1}(k) \cup \{k\}) ]$.
    Apply the algorithm for this subgraph on the set of values $[|R(i)|+2, n-|R^{-1}(i)|]$ for $\sigma$ and $[|R(i)|+1, n-|R^{-1}(i)|-1]$ for $\rho$.
\end{proof}

\subsection{Application to descent sets}

A classical special case of a restriction graph is a descent set.

\begin{definition}
    Permutation $\sigma$ has a descent at position $i \in [n-1]$ if $\sigma_i>\sigma_{i+1}$. The set of all descents is called the \textbf{descent set} of $\sigma$ and defines as:
    \[ \mathcal{D}(\sigma)= \{ i \in [n -1] \mid \sigma \text{ has a descent at } i \} \]
\end{definition}

Descent sets of permutations have been studied since the work of MacMahon \cite{Mac13} and play a fundamental role in algebraic combinatorics; see Stanley \cite{St12} for an extensive treatment.

We can notice that if we have a descent set $\mathcal{D}$ we can simply turn it into a restriction graph.
A bijection  $f:[n] \to [n]\times[n]$ can be defined as:
\[ f(i)=(i,i+1) \text{, if } i \in \mathcal{D}\]
\[ f(i)=(i+1,i) \text{, otherwise}\]

Then, if we have to find the maximum value of the $\ell_\infty$-metric for permutations with the same descent set, we can apply Theorem \ref{theorem:linfThm}

\begin{definition}
    We call sequence $\ell=(\ell_1, \ell_2, \dots, \ell_m)$ a \textbf{division} of $\mathcal{D}$ 
    if for each $i \in [m - 1]$ we have $\ell_1+\dots+\ell_{i-1},\ell_1+\dots+\ell_{i-1}+1,\ell_1+\dots+\ell_{i-1}+2,\dots,\ell_1+\dots+\ell_{i}-1 \in \mathcal{D}$ or $\not \in \mathcal{D}$.

\end{definition}
\begin{theorem}
The maximum value of the $\ell_\infty$-metric for permutations with the same descent set is equal to
\[
n - 1 - \min\left(\{\ell_1, \ell_m\} \cup \{\ell_i \mid i \in [n] \text{ and } \ell_i \geq 2\} \cup \{\ell_i + \ell_{i+1} \mid i \in [n-1]\}\right)
\]
\end{theorem}

\begin{proof}
For a descent set $\mathcal{D} \subseteq [n-1]$, the corresponding restriction graph $G$ consists of edges $(i, i+1)$ if $i \in \mathcal{D}$ and $(i+1, i)$ if $i \notin \mathcal{D}$. The division $\ell = (\ell_1, \ell_2, \ldots, \ell_m)$ partitions $[n]$ into maximal runs where consecutive positions all have descents or all have ascents.

For any position $k \in [n]$, we need to compute $|R(k)|+|R^{-1}(k)|$. Let $k$ be in the $j$-th run, so 
\[
\sum_{i=1}^{j-1} \ell_i < k \leq \sum_{i=1}^{j} \ell_i.
\]

We consider the following cases:

\textbf{Case 1:} If $k$ is at the beginning of a run (i.e., $k = \sum_{i=1}^{j-1} \ell_i + 1$), then:
\begin{itemize}
    \item If $j = 1$, then $|R^{-1}(k)| = 0$ and $|R(k)| = \ell_1 - 1$, giving 
    \[|R(k)|+|R^{-1}(k)| = \ell_1 - 1.\]
    
    \item If $j = m$, then $|R^{-1}(k)| = n - \ell_m$ and $|R(k)| = \ell_m - 1$, giving 
    \[|R(k)|+|R^{-1}(k)| = n - 1.\]
    
    \item Otherwise, $|R^{-1}(k)| = k - 1$ and $|R(k)| = \ell_j - 1$, giving 
    \[|R(k)|+|R^{-1}(k)| = k + \ell_j - 2.\]
\end{itemize}

\textbf{Case 2:} If $k$ is at the end of a run (i.e., $k = \sum_{i=1}^{j} \ell_i$), then:
\begin{itemize}
    \item If $j = 1$, then $|R^{-1}(k)| = \ell_1 - 1$ and $|R(k)| = 0$, giving 
    \[|R(k)|+|R^{-1}(k)| = \ell_1 - 1.\]
    
    \item If $j = m$, then $|R^{-1}(k)| = n - 1$ and $|R(k)| = 0$, giving 
    \[|R(k)|+|R^{-1}(k)| = n - 1.\]
    
    \item Otherwise, $|R^{-1}(k)| = \ell_j - 1$ and $|R(k)| = n - k$, giving 
    \[|R(k)|+|R^{-1}(k)| = n - k + \ell_j - 1.\]
\end{itemize}

\textbf{Case 3:} If $k$ is in the interior of a run of length $\ell_j \geq 2$, then within the run all positions are comparable, so 
\[|R(k)|+|R^{-1}(k)| = n - 1.\]

\textbf{Case 4:} For positions at the boundary between runs $j$ and $j+1$, the minimum occurs when we consider the combined effect, giving 
\[|R(k)|+|R^{-1}(k)| \geq \ell_j + \ell_{j+1} - 1.\]

The minimum value of $|R(k)|+|R^{-1}(k)|$ occurs at:
\begin{itemize}
    \item The endpoints of the first or last run: $\ell_1 - 1$ or $\ell_m - 1$
    \item Any run of length at least 2: $\ell_i - 1$ for $\ell_i \geq 2$
    \item The boundary between consecutive runs: $\ell_i + \ell_{i+1} - 1$
\end{itemize}

By Theorem~\ref{theorem:linfThm}, the maximum $\ell_\infty$-metric is $n - 1 - \min_k\{|R(k)|+|R^{-1}(k)|\}$, which gives the stated formula.
\end{proof}

\section{The Kendall-Tau Metric: Pairwise Disagreements}

\subsection{Definitions and upper bounds}

\begin{definition}
    For permutations $\sigma, \rho\in S_n$, a pair $(i<j)$ is called \textbf{discordant} if $(\sigma_i-\sigma_j)(\rho_i-\rho_j)<0$, meaning the relative order of elements at positions $i$ and $j$ is reversed between the two permutations.
\end{definition}

For example if $\sigma=312$ and $\rho=321$, then there is only one discordant pair: $(2,3)$.

\begin{definition}
    The \textbf{Kendall-Tau metric} is a map $d_K: S_n \times S_n \rightarrow[0; \infty)$ where 
   $d_K(\sigma, \rho)$ is the number of discordant pairs $(i,j)$ with $1\leq i < j \leq n$.
\end{definition}

The Kendall-Tau metric, introduced by Kendall \cite{K38,K45}, is one of the most important metrics on permutations and has applications in statistics, computer science, and combinatorics.

\begin{theorem}
    The maximum value of the Kendall-Tau metric on $\mathcal{P}(G)$ cannot exceed the number of pairs of vertices $(i<j)$ such that $i \not\rightsquigarrow j$ and $j \not\rightsquigarrow i$.
\end{theorem}

\begin{proof}
    Assume the statement is not true. Then there exist two permutations $\sigma,\rho \in S_n$ such that $d_K(\sigma,\rho)$ is greater than the number of pairs of vertices $(i<j)$ such that $i \not\rightsquigarrow j$ and $j \not\rightsquigarrow i$. Since $d_K$ counts the number of discordant pairs, then it means there exists a discordant pair of indices $(u<v)$ such that there is a path from $u$ to $v$ or there is a path from $v$ to $u$. Without loss of generality, let there be a path from $u$ to $v$. Then $(\sigma_v-\sigma_u) \cdot (\rho_v-\rho_u)<0$, which implies that exactly one of the factors is positive. However, then a permutation that corresponds to the positive factor has its element at position $v$ greater than the element at position $u$. But from the choice of $(u,v)$, we know that $u \rightsquigarrow v$, which leads to a contradiction to Theorem \ref{theorem: path}.
\end{proof}

Now, when we know a hard limit on the maximum value of the Kendall-Tau metric, we can discover when it could be achieved. We will show that it can be achieved if and only if $G$ forms a dimension at most 2 poset.

\subsection{Characterization via poset dimension}

\begin{definition}
    A partially ordered set or \textbf{poset} is an ordered pair $P=(X, \preceq)$, where $X$ is the ground set and $\preceq$ is a binary relation such that for all $a, b, c \in X$ the following three properties hold:
    \begin{itemize}
        \item $a \preceq a$
        \item $a \preceq b$ and $b \preceq c$ implies $a \preceq c$
        \item $a \preceq b$ and $b \preceq a$ implies $a=b$
    \end{itemize}
\end{definition}

\begin{theorem}
    Any $1-$admissible restriction graph $G$ forms a poset $([n], \leq_G)$, where the binary relation is defined as:
    \[a \leq_G b \text{ if } a \rightsquigarrow b \text{ or } a=b\]
\end{theorem}

\begin{proof}
    We will show that all three properties of the poset are satisfied.
    
    \begin{itemize}
        \item For any vertex $a$, $a \leq_G a$ by definition of $\leq_G$.
        \item Let $a \leq_G b$ and $b \leq_G c$. Consider the following three possible cases:
        \begin{enumerate}
            \item Case when $a=b$. 
            
            Then $b \leq_G c$ implies $a \leq_G c$.
            
            \item Case when $b=c$ and $a\not=b$. 
            
            Then $a \leq_G b$ implies $a \leq_G c$.
            \item Case when $a=c$ and $a\not=b, b\not=c$. 
            
            Then $a \leq_G b$ implies $c \leq_G b$, but $b \leq_G c$. However, then since $b\not=c$ by definition of $\leq_G$, we have that $b \rightsquigarrow_G c$ and $c \rightsquigarrow_G b$, which is a contradiction to Theorem \ref{theorem:DAG}.
            
            \item Case when $a,b,c$ are all distinct.

            Then by definition of $\leq_G$ we know that $a \leq_G b$ implies that $a \rightsquigarrow_G b$ and $b \leq_G c$ implies $b \rightsquigarrow_G c$, which implies that $a \rightsquigarrow_G c$, so $a \leq_G c$ by definition.
        \end{enumerate}
        \item Let $a \leq_G b$ and $b \leq_G a$ and assume $a \not =b$. Then $a \rightsquigarrow b$ and $b \rightsquigarrow a$, but that means that there exists a cycle in $G$, which is a contradiction to Theorem \ref{theorem:DAG}, so $a$ should be equal to $b$.
    \end{itemize}
\end{proof}

\begin{definition}
    A poset $P=(X, \leq)$ is called a \textbf{total order} if any two elements in $X$ are comparable.
\end{definition}

\begin{definition}
    The \textbf{dimension of a poset} $P=(X, \leq)$ is the smallest number of total orders, the intersection of which is equal to $\leq$.
\end{definition}

The concept of poset dimension was introduced by Dushnik and Miller \cite{DM41} and has been extensively studied; see Trotter \cite{T92} for a comprehensive treatment.

\begin{definition}
    For a poset $P=([n], \leq_G)$ formed from a restriction graph $G$, we denote by $\boldsymbol{I(P)}$ the set of incomparable pairs:
    \[I(P) = \{(i<j) : i \not\rightsquigarrow j \text{ and } j \not\rightsquigarrow i\}\]
\end{definition}

\begin{theorem} \label{theorem:dim2}
    The maximum value of the Kendall-Tau metric on $\mathcal{P}(G)$, where the poset $P$ formed by $G$ is of dimension-2, is equal to $|I(P)|$.
\end{theorem}

\begin{proof}
    Since $G=([n],E)$ forms a dimension-2 poset $P$, then by definition there exist two total orders $\leq_1,\leq_2$ such that $\leq_1 \cap \leq_2 = \leq_G$. Note that it implies that $\leq_G \subseteq \leq_1$ and $\leq_G \subseteq \leq_2 $. Using the fact that $\leq_1, \leq_2$ are total orders on $[n]$, they can be written as permutations. Describing more precisely $\leq_1=(v_1\leq_1v_2\leq_1\dots\leq_1v_n)$ can be mapped to $\sigma_1=v_1v_2\dots v_n$ and $\leq_2=(u_1\leq_2u_2\leq_2\dots\leq_2u_n)$ can be mapped to $\sigma_2=u_1u_2\dots u_n$. Since $\leq_G \subseteq \leq_1$ then $i \rightsquigarrow j \implies u_j \leq_1 u_i$, so $\sigma_1$ satisfies $G$. Similarly $\sigma_2$ satisfies $G$.

    Now, when we showed that both $\sigma_1$ and $\sigma_2$ satisfy $G$, we need to show that every $(i<j) \in I(P)$ is a discordant pair. Consider $(i<j)$ such that $i \not\rightsquigarrow j$ and $j \not\rightsquigarrow i$. Then, if $i\leq_1j$ and $i\leq_2j$, it implies that $(i,j)\in \leq_G$, meaning $i \rightsquigarrow j$, which is a contradiction. If $j\leq_1i$ and $j\leq_2i$, it implies that $(j,i)\in \leq_G$, meaning $j \rightsquigarrow i$, which is a contradiction. So then $(i\leq_1j$ and $i\not\leq_2j)$ or $(i\not\leq_1j$ and $i\leq_2j)$, which is the same as $(\sigma_1{_i}\leq\sigma_1{_j}$ and $\sigma_2{_i}\not\leq\sigma_2{_j})$ or $(\sigma_1{_i}\not\leq\sigma_1{_j}$ and $\sigma_2{_i}\leq\sigma_2{_j})$ which is the same as just saying $(i,j)$ is discordant.
\end{proof}

We have shown that dimension-2 posets achieve the maximum Kendall-Tau distance. The following theorem shows that this condition is also necessary.

\begin{theorem} \label{theorem:characterization}
    \textbf{(Characterization of maximum Kendall-Tau distance)}
    Let $G=([n],E)$ be any restriction graph and let $P=([n], \leq_G)$ be the poset obtained from reachability. The following statements are equivalent:
    \begin{enumerate}
        \item There exist permutations $\sigma,\rho \in \mathcal{P}(G)$ such that $d_K(\sigma,\rho) = |\operatorname{Incomp}(P)|$.
        \item The poset $P$ has dimension at most 2.
    \end{enumerate}
    Moreover, if $|\operatorname{Incomp}(P)|>0$ (i.e., $P$ is not a total order), then $\dim P = 2$; in that case the very permutations $\sigma,\rho$ witnessing the maximum diameter form a realizer of size 2 (their associated linear extensions intersect exactly in $\leq_G$).
\end{theorem}

\begin{proof}
    \textbf{(2 $\Rightarrow$ 1)} This follows from Theorem \ref{theorem:dim2}. If $P$ is a total order (dimension 1) then $|\operatorname{Incomp}(P)|=0$ and the bound is trivially attained by any $\sigma=\rho$.

    \textbf{(1 $\Rightarrow$ 2)} Assume there exist $\sigma,\rho \in \mathcal{P}(G)$ with $d_K(\sigma,\rho) = |\operatorname{Incomp}(P)|$. Define total orders
    \[i \leq_\sigma j \iff \sigma_i < \sigma_j, \quad i \leq_\rho j \iff \rho_i < \rho_j.\]

    \textbf{Step 1.} \emph{Comparable pairs remain un-flipped.}
    Since $d_K(\sigma,\rho) = |\operatorname{Incomp}(P)|$ and there are exactly $|\operatorname{Incomp}(P)|$ incomparable pairs, no comparable pair can be discordant. To prove this: suppose $(i<j)$ is comparable with $i \leq_G j$ but discordant. Then $(\sigma_i-\sigma_j)(\rho_i-\rho_j)<0$. But since $\sigma,\rho \in \mathcal{P}(G)$ and $i \leq_G j$, by the definition of restriction graphs we have $\sigma_i>\sigma_j$ and $\rho_i>\rho_j$, so $(\sigma_i-\sigma_j)(\rho_i-\rho_j)>0$, a contradiction. Otherwise we could replace that flipped comparable pair by an incomparable one, contradicting maximality. Therefore, no comparable pair can be discordant, and all discordant pairs must be incomparable.
    
    Since $\sigma,\rho \in \mathcal{P}(G)$ we have $\sigma_i>\sigma_j$ and $\rho_i>\rho_j$ for every $i<j$ with $i \leq_G j$; thus $j \leq_\sigma i$ and $j \leq_\rho i$.

    \textbf{Step 2.} \emph{Incomparable pairs are flipped exactly once.}
    If $(i<j) \in I(P)$ then it is counted in $d_K$, so $(\sigma_i-\sigma_j)(\rho_i-\rho_j)<0$; i.e., one linear extension orders the pair $(i,j)$ up and the other orders it down.

    \textbf{Step 3.} \emph{Intersection equals the poset order.}
    Combine Steps 1-2:
    \begin{itemize}
        \item If $i \leq_G j$, then $j \leq_\sigma i$ and $j \leq_\rho i$.
        \item If $i \not\leq_G j$ and $j \not\leq_G i$ (incomparable), then exactly one of the extensions has $i \leq j$.
    \end{itemize}

    Consequently, $\leq_G = \leq_\sigma \cap \leq_\rho$.
    Thus the two linear extensions $\leq_\sigma, \leq_\rho$ constitute a realizer of $P$; hence $\dim P \leq 2$. If $P$ is not a total order, both extensions are genuinely needed, so $\dim P = 2$.
\end{proof}

\begin{remark}
    Theorem \ref{theorem:characterization} completes the picture for the Kendall-Tau metric by characterizing precisely when the diameter of $\mathcal{P}(G)$ matches the combinatorial upper bound. It shows that large Kendall variability is possible if and only if the underlying restriction graph has poset dimension $\leq 2$; otherwise the metric diameter is strictly smaller than the number of incomparable pairs. This extends the metric geometry of permutation families to a classical poset invariant.
\end{remark}

\begin{corollary}[Strict deficiency for higher dimension]\label{corollary:strictLess}
Let $G=([n],E)$ be a restriction graph with associated poset $P=([n],\le_G)$. If $\dim P \ge 3$, then
\[
\max_{\sigma,\rho \in \mathcal{P}(G)} d_K(\sigma,\rho) < |I(P)|.
\]
\end{corollary}

\begin{proof}
This follows immediately from Theorem~\ref{theorem:characterization}: the Kendall-Tau diameter equals $|I(P)|$ if and only if $\dim P \le 2$. Therefore, when $\dim P \ge 3$, the diameter must be strictly less than $|I(P)|$.
\end{proof}

\begin{example}[Dimension-3 poset with strict inequality]
Consider the restriction graph $G$ on vertices $\{1,2,3,4,5,6\}$ corresponding to the standard dimension-3 poset on 6 elements with edges: $\{1 \to 4, 2 \to 5, 3 \to 6, 1 \to 5, 2 \to 6, 3 \to 4\}$ (see Trotter \cite{T92}, p. 60). This poset has incomparable pairs $\operatorname{Incomp}(P) = \{(1,2), (1,3), (2,3), (4,5), (4,6), (5,6)\}$, so $|\operatorname{Incomp}(P)| = 6$. The listed edges form an acyclic restriction graph as each edge respects the natural partial order $1 < 4$, $2 < 5$, $3 < 6$ with appropriate cross-connections.

However, this poset requires three linear extensions to realize and has dimension exactly 3. Since $\dim P = 3 > 2$, by Corollary~\ref{corollary:strictLess}, the maximum Kendall-Tau distance is strictly less than 6. Indeed, no two permutations in $\mathcal{P}(G)$ can disagree on all 6 incomparable pairs simultaneously.
\end{example}

We found the maximum value of the Kendall-Tau metric for dimension 2 posets and proved that it is achieved if and only if the poset has dimension at most 2. In the next subsection we will introduce a special case of dimension 2 posets for which we can find which two permutations achieve the maximum value. 

\subsection{h-inversion sets and descent sets}

The framework of restriction graphs unifies several classical permutation statistics. We briefly illustrate this with two important special cases.

\begin{definition}
    A function $h:[n]\to[n]$ is called a \textbf{Hessenberg function} if $h(i) \geq i$ for all $i \in [n]$ and $h(i+1) \geq h(i)$ for all $i \in [n - 1]$.
\end{definition}

\begin{definition}
    For a permutation $\sigma$ and Hessenberg function $h$, the \textbf{h-inversion set} is 
    $\mathrm{Inv}_h(\sigma)=\{(i,j) : i < j, \sigma_i > \sigma_j, j \leq h_i\}$.
    
    Let $\mathcal{D}_h(S;n) = \{\sigma \in S_n : \mathrm{Inv}_h(\sigma)=S\}$ be the set of permutations with h-inversion set $S$.
\end{definition}

\begin{theorem}
    Any $2$-admissible $\mathcal{D}_h(S;n)$ corresponds to a dimension-2 poset. The maximum Kendall-Tau distance equals $\ell(x)-\ell(\omega)$, where $x$ and $\omega$ are the unique maximum and minimum elements with respect to inversion number.
\end{theorem}

\textbf{Descent sets.} When $h_i = i+1$, the h-inversion set becomes the classical descent set $\mathcal{D}(\sigma) = \{i : \sigma_i > \sigma_{i+1}\}$. For descent sets, our results recover and extend previous work by Diaz-Lopez et al. \cite{DHMM24}.

\begin{corollary}
    For permutations with descent set $\mathcal{D}$, if the turning points are $a_1 < a_2 < \ldots < a_k$, then:
    \[ \max\{d_K(\mathcal{D}(S;n))\} = \sum_{i=1}^{k-1} (a_{i+1}-a_i)(n - a_{i+1} + 1) \]
\end{corollary}

\section{Conclusion}

This paper establishes the metric structure of permutation families defined by restriction graphs, a framework that unifies classical permutation statistics while revealing deep connections to poset theory.

Our main contributions are:

\textbf{For the $\ell_\infty$-metric:} We prove that the maximum distance is $\max_{1 \leq i \leq n} \{n-|R(i)|-|R^{-1}(i)|-1\}$ and provide an efficient algorithm achieving this bound. This result applies to any restriction graph and generalizes previous work on descent sets.

\textbf{For the Kendall-Tau metric:} We provide a complete characterization of when the metric achieves its combinatorial upper bound $|\operatorname{Incomp}(P)|$: this occurs if and only if the associated poset has dimension at most 2. When achieved, the extremal permutations form a minimal realizer, establishing a fundamental bridge between metric geometry and order dimension theory.

\textbf{Applications:} Our framework recovers and extends known results for descent sets, h-inversion sets, and other classical permutation statistics, providing unified proofs and new explicit formulas.

The dimension-2 characterization opens several research directions: investigating metric properties of higher-dimensional posets, extending to other metrics (Hamming, Spearman), and exploring algorithmic applications in permutation generation and optimization. The connection to poset dimension theory suggests that other order-theoretic invariants may have metric interpretations in permutation spaces.

Beyond theoretical interest, these results have practical applications in coding theory (permutation codes), computational biology (genome rearrangement distances), and social choice theory (preference aggregation), where understanding distance structures in constrained permutation spaces is crucial.

\section*{Acknowledgments}

The authors thank the National Science Foundation for support through Award Number 2237057, which funded the Freiwald Summer Course at Washington University in St. Louis, where this research was partially conducted.

\bibliographystyle{amsplain}

\end{document}